\documentclass[12pt]{amsart}
\usepackage{amssymb,amscd}
\usepackage{verbatim}

\usepackage{amsmath,amssymb,graphicx,mathrsfs}   
\usepackage[colorlinks=true,allcolors = blue]{hyperref} 

\let\frak\mathfrak
\let\Bbb\mathbb

\def\>{\relax\ifmmode\mskip.666667\thinmuskip\relax\else\kern.111111em\fi}
\def\<{\relax\ifmmode\mskip-.333333\thinmuskip\relax\else\kern-.0555556em\fi}
\def\vsk#1>{\vskip#1\baselineskip}
\def\vv#1>{\vadjust{\vsk#1>}\ignorespaces}
\def\vvn#1>{\vadjust{\nobreak\vsk#1>\nobreak}\ignorespaces}

  \let\ssize\scriptstyle
\let\sssize\scriptscriptstyle

\let\Medskip\medskip
\def\medskip{\par\Medskip}
\let\Bigskip\bigskip
\def\bigskip{\par\Bigskip}

\let\Maketitle\maketitle
\def\maketitle{\Maketitle\thispagestyle{empty}\let\maketitle\empty}

\newtheorem{thm}{Theorem}[section]
\newtheorem{cor}[thm]{Corollary}
\newtheorem{lem}[thm]{Lemma}
\newtheorem{prop}[thm]{Proposition}

\theoremstyle{definition}                                  
\numberwithin{equation}{section}

\theoremstyle{definition}
\newtheorem*{rem}{Remark}
\newtheorem*{example}{Example}

\let\mc\mathcal
\let\nc\newcommand

\let\la\lambda

\let\phi\varphi
\let\si\sigma

\let\om\omega
\let\Om\Omega

\let\der\partial

\let\ox\otimes

\let\geq\geqslant

\let\leq\leqslant

\let\on\operatorname
\let\bi\bibitem
\let\bs\boldsymbol

\def\C{{\mathbb C}}
\def\Z{{\mathbb Z}}

\def\F{{\mathbb F}}   

\def\+#1{^{\{#1\}}}

\def\beq{\begin{equation}}
\def\eeq{\end{equation}}
\def\be{\begin{equation*}}
\def\ee{\end{equation*}}

\nc{\bea}{\begin{eqnarray*}}
\nc{\eea}{\end{eqnarray*}}
\nc{\bean}{\begin{eqnarray}}
\nc{\eean}{\end{eqnarray}}

\let\ga\gamma

\nc{\Il}{{\mc I_{\bs\la}}}
\nc{\bla}{{\bs\la}}
\nc{\Fla}{\F_\bla}
\nc{\tfl}{{T^*\Fla}}
\nc{\GL}{{GL_n(\C)}}
\nc{\GLC}{{GL_n(\C)\times\C^*}}

\let\sd s 

\def\ddk_#1{\kk_{#1}\<\>\frac\der{\der\<\>\kk_{#1}}}

\def\bul{\mathbin{\raise.2ex\hbox{$\sssize\bullet$}}}
\def\intt{\mathchoice
{\mathop{\raise.2ex\rlap{$\,\,\ssize\backslash$}{\intop}}\nolimits}
{\mathop{\raise.3ex\rlap{$\,\sssize\backslash$}{\intop}}\nolimits}
{\mathop{\raise.1ex\rlap{$\sssize\>\backslash$}{\intop}}\nolimits}
{\mathop{\rlap{$\sssize\<\>\backslash$}{\intop}}\nolimits}}

\let\kk q 
\let\cc c

\let\Ko K

\def\GZ/{Gelfand-Zetlin}
\def\KZ/{{\slshape KZ\/}}
\def\qKZ/{{\slshape qKZ\/}}
\def\XXX/{{\slshape XXX\/}}

\def\Sym{\on{Sym}}

\nc{\A}{{\mc C}}

\def\Sing{{\on{Sing}\,}}
\def\sll{{\frak{sl}}}

\def\slt{{\frak{sl}_2}}

\def\Ik{{\mc J_r}}

\nc{\hsl}{\widehat{{\frak{sl}_2}}}

\nc{\BC}{{ \mathbb C}}
\nc{\lra}{\longrightarrow}
\nc{\CO}{{\mathcal{O}}}
\nc{\BZ}{{ \mathbb Z}}
\nc{\hfn}{\hat{\frak{n}}}

\begin{document}

\hrule width0pt
\vsk->

\title[On the number of $p$-hypergeometric solutions]
{On the number of $p$-hypergeometric solutions
\\
 of \KZ/ equations}

\author[Alexander Varchenko]
{Alexander Varchenko}

\maketitle

\begin{center}
{\it $^{\star}$ Department of Mathematics, University
of North Carolina at Chapel Hill\\ Chapel Hill, NC 27599-3250, USA\/}

\vsk.5>
{\it $^{ \star}$ Faculty of Mathematics and Mechanics, Lomonosov Moscow State
University\\ Leninskiye Gory 1, 119991 Moscow GSP-1, Russia\/}


\end{center}

\vsk>
{\leftskip3pc \rightskip\leftskip \parindent0pt \Small
{\it Key words\/}:  \KZ/ equations; master polynomials; $p$-hypergeometric solutions.

\vsk.6>
{\it 2020 Mathematics Subject Classification\/}: 11D79 (32G34, 33C60, 33E50)
\par}

{\let\thefootnote\relax
\footnotetext{\vsk-.8>\noindent
$^\star\<${\sl E\>-mail}:\enspace anv@email.unc.edu,
supported in part by NSF grant DMS-1954266
}}

\begin{abstract}

It is known that solutions of the \KZ/ equations can be written in the form of multidimensional hypergeometric
integrals. In 2017 in a joint paper of the author with V.\,Schechtman the construction of 
hypergeometric solutions was modified,
and solutions of the \KZ/ equations modulo a prime number $p$ were constructed. 
These solutions modulo $p$, called the $p$-hypergeometric 
solutions, are polynomials with integer coefficients.
A general problem is to determine the number of independent $p$-hypergeometric solutions
and understand the meaning of that number. 

In this paper we consider the \KZ/ equations associated with the space of singular vectors of weight $n-2r$ in the tensor power
$W^{\ox n}$ of the vector representation of $\frak{sl}_2$. In this case the hypergeometric solutions of the \KZ/ equations 
are given by $r$-dimensional hypergeometric integrals.  We consider the 
 module of  the corresponding $p$-hypergeometric solutions, determine its rank, 
and show that the rank equals the dimension of the space of suitable square integrable differential $r$-forms. 

\end{abstract}

{\small\tableofcontents\par}

\setcounter{footnote}{0}
\renewcommand{\thefootnote}{\arabic{footnote}}

\section{Introduction}

The Knizhnik-Zamolodchikov (\KZ/)
 equations is a system of differential equations for conformal  blocks in conformal field theory, see \cite{KZ}.
Versions of \KZ/ equations appear in mathematical physics, representation theory, enumerative geometry, algebraic geometry,
theory of special functions,  see for example \cite{EFK, MO, B}.

It is known that solutions of the \KZ/ equations can be written in the form of multidimensional hypergeometric
integrals, see \cite{SV1}.  Relatively recently, the construction of these
hypergeometric solutions was modified,
and solutions of the \KZ/ equations modulo a prime number $p$ were constructed  in \cite{SV2}. 
These solutions modulo $p$, called the $p$-hypergeometric solutions, are vector-valued 
polynomials with integer coefficients.
The general problem  is 
to understand how these polynomials reflect the remarkable properties  of the \KZ/ equations and their
hypergeometric solutions.

\vsk.2>

In this paper we address the particular problem to determine the number of independent $p$-hypergeometric solutions and understand the meaning of that number. 
We consider the \KZ/ equations associated with the space of singular
 vectors of weight $n-2r$ in the tensor power
$W^{\ox n}$ of the vector representation of $\frak{sl}_2$. In this case the hypergeometric solutions of the \KZ/ equations 
are given by $r$-dimensional hypergeometric integrals.  We consider the 
 module of  the corresponding $p$-hypergeometric solutions, determine its rank, 
and show that the rank equals the dimension of the space of suitable square integrable differential $r$-forms.

The KZ equations depend on a parameter $q\in\C^\times$. In this paper we assume that 
$q$ is a prime number less than $p$, and the pair $(p,q)$ satisfies
 certain conditions (the pair is of type 1).
 
 \vsk.2>
 On $p$-hypergeometric solutions and, more generally, on the solutions of the KZ equations modulo $p^s$ see
 \cite{SlV, V4, V5, V6, V7, V8, VZ1, VZ2}.

\vsk.2>

In Section \ref{sec 2} we define the KZ equations, recall the construction of solutions 
in the form of hypergeometric integrals and the construction of $p$-hypergeometric solutions. We also
introduce the module of $p$-hypergeometric solutions. 
The main result of the paper is Theorem \ref{thm li},
formulated in Section \ref{sec 3} and proved in Section \ref{sec 4}.
Theorem \ref{thm li} 
determines the rank of the module  of $p$-hypergeometric solutions.

In Section \ref{sec 5} 
we construct a suitable Cartier map, which relates 
 the hypergeometric solutions and  $p$-hypergeometric solutions. As a result
 of this construction, 
we interpret the rank of the module of $p$-hypergeometric solutions as the dimension
of some vector space of square integrable differential $r$-forms on $\Bbb P^r$,
the $r$-th direct power of the complex projective line.
Previously this result was known for $r=1$, see \cite{SlV}.

\medskip
\noindent
{\bf Acknowledgements.}
The author thanks Alexey Slinkin for useful discussions.

\section{$\sll_2$ \KZ/ equations}
\label{sec 2}

\subsection{Definition of equations}

Consider the  complex Lie algebra $\slt$ with generators $e,f,h$ and relations $[e,f]=h,\, [h,e]=2e, \,[h,f]=-2f$.
Consider the complex vector space $W$ with basis $w_1,w_2$ and the $\slt$-action,
\bea
e=\begin{pmatrix}
0 & 1
\\
0 & 0 & 
\end{pmatrix}, \qquad
f=\begin{pmatrix}
0 & 0
\\
1 & 0 & 
\end{pmatrix},
\qquad
h=\begin{pmatrix}
1 & 0
\\
0 & -1 & 
\end{pmatrix}.
\eea 
The $\slt$-module $W^{\ox n}$ has a basis labeled by subsets $J\subset\{1,\dots,n\}$,
\bea
V_J= w_{j_1}\ox \dots\ox w_{j_n},
\eea
where $j_i =1$ if $j_i\notin J$  and
$j_i =2$ if $j_i\in J$.

\vsk.2>

Consider the weight decomposition of $W^{\ox n}$ 
into eigenspaces of $h$,
\\
  $W^{\ox n} = \sum_{r=0}^n W^{\ox n}[n-2r]$. 
The vectors $V_J$ with $|J|=r$ form a  basis of $W^{\ox n}[n-2r]$.
Denote $\Ik$ the set of all $r$-element subsets of $\{1,\dots,n\}$.

Define the space of singular vectors of weight $n-2r$,
\bea
\Sing W^{\ox n}[n-2r] = \{ w\in W^{\ox n}[n-2r]\mid ew=0\}.
\eea
This space is nonempty if and only if $n\geq 2r$. We assume that $n\geq 2r$.
Then
\bea
\dim \Sing W^{\ox n}[n-2r] = \dim W^{\ox n}[n-2r] - \dim W^{\ox n}[n-2r+2]
=\binom{n}{r} - \binom{n}{r-1}.
\eea
Let $w=\sum_{J\in\Ik} c_JV_J\in W^{\ox n}[n-2r]$. Then $w\in \Sing W^{\ox n}[n-2r]$,
if and only if its coefficients satisfy the system of linear equations labeled
by $r-1$-element subsets  $K\subset\{1,\dots,n\}$,
\bean
\label{ev=0} 
\sum_{j\notin K}  c_{K\cup \{j\}}\,=\,0.
\eean
Define the Casimir element 
$\Om =  \frac12 h\ox h + e\ox f+f\ox e\, \in\, \slt\ox\slt$,
and the linear operators on $W^{\ox n}$ depending on parameters $z=(z_1,\dots,z_n)$,
\bea
 H_m(z) = \sum_{j=1\atop j\ne m}^n\frac{ \Om_{mj}}{z_m-z_j}\,,
\qquad m=1,\dots,n,
\eea
where 
$ \Om_{mj}:W^{\ox n}\to W^{\ox n}$ is the Casimir operator acting in the $m$th and $j$th tensor factors.
The operators $ H_m(z)$ are called the Gaudin Hamiltonians.  Denote
\bea
\der_m =\frac{\der}{\der z_m}\,,\qquad m=1,\dots,n.
\eea
For any nonzero number $q\in\C^\times$ and  $1\leq m,l\leq n$, we have
 \bean
 \label{Flat}
 [q\der_m -   
 H_m(z_1,\dots,z_n), q\der_l-   H_l (z_1,\dots,z_n) ]
=0,
\eean 

\noindent
and for any $x\in\slt$ and $m=1,\dots,n,$    we have
\bean
\label{h inv}
[ H_m(z_1,\dots,z_n), x\ox 1\ox\dots\ox 1+\dots + 1\ox\dots\ox 1\ox x] =0.
\eean
The system of differential equations 
\bean
\label{kz}
(q\der_m -     H_m(z_1,\dots,z_n))
 \tilde I(z_1,\dots,z_n), 
\qquad  m = 1, \dots , n,
\eean
on  a $W^{\ox n}$-valued function
$\tilde I(z_1, \dots, z_n)$ is called the \KZ/ equations, see \cite{KZ, EFK}.

\vsk.2>

By property \eqref{h inv} the Gaudin Hamiltonians  preserve every space 
$\Sing W^{\ox n}[n-2r]$.  Hence the system of  \KZ/ equations can be considered with values in
any particular space
\\
 $\Sing W^{\ox n}[n-2r]$.

\subsection{Gauge transformation} 

If $\tilde I(z)$ satisfies the \KZ/ equations \eqref{kz}, then the function $I(z)$, defined by 
\bean
\label{tilde I}
\tilde I(z) = I(z) \prod_{1\leq i<j\leq n}(z_i-z_j)^{1/2q}\,,
\eean
satisfies the equations
\bean
\label{KZ}
\Big(q\der_m - \sum_{j=1\atop j\ne m}^n\frac{ \Om_{mj}-1/2}{z_m-z_j}\Big) 
I(z_1,\dots,z_n), 
\qquad  m = 1, \dots , n,
\eean
which we also call the \KZ/ equations.

\vsk.2>

The linear operator $\Om -\frac12 : W^{\ox 2} \to W^{\ox 2}$ acts  as follows:
\bean
\label{Oa}
&&
w_1\ox w_1 \mapsto 0, \qquad 
\phantom{aaaaaaaaaaaaaaaa}
w_2\ox w_2 \mapsto 0,
\\
\notag
&&
w_1\ox w_2 \mapsto - w_1\ox w_2 + w_2\ox w_1,
\qquad 
w_2\ox w_1 \mapsto  w_1\ox w_2 - w_2\ox w_1.
\eean

\vsk.2>

We shall consider the \KZ/ equations \eqref{KZ} with values in $\Sing W^{\ox n}[n-2r]$
over the field of complex numbers, and we shall also consider 
the \KZ/ equations \eqref{KZ} modulo $p$, where $p$ is a
prime number.

\subsection
{Solutions of \KZ/ equations over $\C$}
  \label{Sol in C}

Denote $t=(t_1,\dots,t_k)$.
Define the   master function
\bean
\label{Master}
\Phi(t,z)
= 
\prod_{1 \leq i \leq j \leq r}  (t_i-t_j)^{2/q}
\prod_{s=1}^{n} \prod_{i=1}^{r} (t_i-z_s)^{-1/q}.
\eean
For any function  $F(t_1, \dots , t_r)$, denote
$\Sym_t  [F(t_1, \dots , t_r)]  = \sum_{\sigma \in S_r}\!
F(t_{\sigma_1}, \dots , t_{\sigma_r})$ .
For $J=\{j_1,\dots,j_r\}\in \Ik$ define the weight function
\bea
W_J(t,z)  =            
\Sym_t \Big[ \prod_{i=1}^{r}  \frac{1}{t_{i} - z_{j_i}} \Big]  .
\eea
For example,
\bea
W_{\{3\}} =  \frac{1}{t_1 - z_3} ,
\qquad
W_{\{4,5\}} =   \frac{1}{t_1-z_4} \frac{1}{t_2-z_5}
+ \frac{1}{t_2-z_4}  \frac{1}{t_1-z_5}\,  .
\eea
The function
\bean
\label{vW}
W(t,z)=\sum_{J\in \Ik}W_J(t,z) V_J
\eean
is  called the  $W^{\otimes n}[n -2r]$-weight vector-function.

\vsk.2>

Consider the  $W^{\otimes n}[n -2r]$-valued
function
\bean
\label{intrep}
I^{(\gamma)}(z_1, \dots , z_n) =  \int_{\ga(z)}  \Phi(t,z)\,
W(t,z)\, dt_1 \wedge \dots
   \wedge dt_r ,
\eean
where $ \gamma(z)$ in $\{z\} \times \C^r_t$ is a horizontal family of
$r$-dimensional cycles of the twisted homology defined by the multivalued
function $\Phi(t,z)$, see \cite{CF,SV1,DJMM,V1,V3}.
The cycles $\gamma(z)$ are $r$-dimensional analogs of Pochhammer double loops.

\begin{thm}
\label{thm s}
The function  $I^{(\gamma)} (z)$ takes values in
 $ \Sing W^{\otimes n}[n-2r]$
and satisfies the \KZ/ equations \eqref{KZ}.
\end{thm}

This theorem and  its generalizations can be found, for example,  in \cite {CF, DJMM, SV1}.

\vsk.2>

The solutions in Theorem \ref{thm s}  are called the  
hypergeometric solutions of the \KZ/ equations.

\begin{thm}
[{\cite[Theorem 12.5.5]{V1}}]
\label{thm alls}
If $q\in\C^\times$ is generic, then all solutions of the \KZ/ equations \eqref{KZ} have this form.

\end{thm}

For special values of the parameter $q$ the space of
the  hypergeometric solutions of the \KZ/ equations may span only 
a proper subspace of the space of all solutions,
see in \cite{FSV1, FSV2} the discussion of 
the relations of this subspace of solutions and conformal blocks in conformal field theory.

\subsection{$p$-integrals}

Let $p$ be a positive integer. Let 
$f(t_1,\dots,t_k)=\sum_{d_1,\dots,d_k} c_{d_1,\dots,d_k} t_1^{d_1}\dots t_k^{d_k}$\,
be a polynomial. Let $l=(l_1,\dots,$ $l_r)\in \Z_{>0}^r$. The coefficient
$c_{l_1p-1,\dots,l_rp-1}$ will be  called the $p$-integral of $f(t_1,\dots,t_k)$
over the cycle $\{l_1,\dots,l_r\}_p$
and is denoted by 
\bea
 \int_{\{l_1,\dots,l_r\}_p} f(t_1,\dots,t_k) \,dt_1\dots dt_r\,.
\eea

\subsection{\KZ/ equations modulo $p$} 
\label{sec 2.5}

Let $p$ and $q$ be prime numbers, $p>q$.
We consider the system of \KZ/ equations \eqref{KZ} with parameter $q$ modulo $p$.
Namely, we look for $W^{\otimes n}[n-2r]$-valued polynomials in $z_1,\dots,z_n$
with integer coefficients, which satisfy system \eqref{KZ} modulo $p$
and with values in the subspace $\Sing W^{\otimes n}[n-2r]$ modulo $p$.
More precisely, let $I(z_1,\dots,z_n) = \sum_{J\in\Ik} I_J(z_1,\dots,z_n) V_J$ with 
$I_J(z_1,\dots,z_n)\in \Z[z_1,\dots,z_n]$. We request that 
\begin{itemize}
\item
for any $m=1,\dots,n$, the rational function
$(q\der_m - H_m(z))I(z)$ can be written as a ratio 
of two polynomials with integer coefficients  such that the denominator
is nonzero modulo $p$, while the numerator is zero modulo $p$;
\item
the coefficients $I_J(z_1,\dots,z_n)$ of the polynomial
$I(z)$ satisfy equations \eqref{ev=0} modulo $p$.
\end{itemize}

\vsk.2>

A construction of such solutions was presented in \cite{SV2}.

\vsk.2>

Let $M$, $c$ be the least positive integers such that 
\bean
\label{Mab}
M\equiv -\frac{1}q,
\qquad
c \equiv \frac2q \ \ 
\pmod{p}.
\eean
Let $t=(t_1,\dots,t_r)$, $z=(z_1,\dots,z_n)$.  
Define the master polynomial, 
\bean
\label{Phi p}
\Phi_p(t,z) = \prod_{1\leq i<j\leq r}(t_i-t_j)^c\prod_{i=1}^r\prod_{s=1}^n(t_i-z_s)^{M}.
\eean
Recall the weight vector-function $W(t,z)$ in \eqref{vW}. The function
$\Phi_p(t,z) W(t,z)$ is a $W^{\otimes n}[n -2r]$-valued polynomial in $t,z$.
Let $(l_1,\dots,l_r)\in \Z_{>0}^r$. 
Denote
\bean
\label{def I}
I^{(l_1,\dots,l_r)}(z) = \int_{\{l_1,\dots,l_r\}_p} \Phi_p(t,z) W(t,z) \,dt_1\dots dt_r\,.
\eean
This is a $W^{\otimes n}[n -2r]$-valued polynomial in $z$.

\begin{thm} [\cite{SV2}]
\label{thm sv2} 
For any positive integers $(l_1,\dots,l_r)$, the polynomial
$I^{(l_1,\dots,l_r)}(z)$ is a solution of the \KZ/ equations modulo $p$ with values in
$\Sing W^{\ox n}[n-2r]$ modulo $p$.

\end{thm}

We call the polynomials
$I^{(l_1,\dots,l_r)}(z)$ the {\it $p$-hypergeometric} solutions. 

\vsk.2>
In this paper we determine the number of independent
$p$-hypergeometric solutions under the assumption that $(p,q)$ is a pair of type 1.

\begin{rem}
The symmetric group $S_n$ acts on the polynomials in $z_1,\dots,z_n$ with values in $W^{\ox n}$
by permuting the tensor factors and the variables simultaneously,
\bea
\si (p(z_1,\dots,z_n)\, v_1\ox\dots\ox v_n) = 
(p(z_{\si(1)},\dots,z_{\si(n)})\, v_{\si^{-1}(1)}\ox\dots\ox v_{\si^{-1}(n)},
\qquad \si\in S_n\,.
\eea

It is clear that for any positive integers $(l_1,\dots,l_r)$, the polynomial
$I^{(l_1,\dots,l_r)}(z)$ is symmetric with respect to the $S_n$-action.

\end{rem}

\subsection{Module of $p$-hypergeometric solutions}
\label{sec dM}

Denote $\Z[z^p]=\Z[z_1^p,\dots,z_{n}^p]$. 
Let 
\bean
\label{Def M_M}
\mathcal{M}\,=\,\Big\{ \sum_{l_1,\dots,l_r} f_{l_1,\dots,l_r}(z) I^{(l_1,\dots,l_r)}(z) \mid  f_{l_1,\dots,l_r}(z)\in\Z[z^p]\Big\},
\eean
be the $\Z[z^p]$-module generated by the $p$-hypergeometric
solutions $I^{(l_1,\dots,l_r)}(z)$ of Theorem  \ref{thm sv2}.
Every element of $\mc M$ is a solution of the \KZ/ equations modulo $p$ with values in
$\Sing W^{\ox n}[n-2r]$ modulo $p$. Indeed  the \KZ/ equations \eqref{KZ} are linear and
$\frac{\der z_i^p}{\der z_j} \equiv 0$ (mod $p$)  for all $i,j$.

\vsk.2>
We say that two elements $I, I' \in \mc M$ are equivalent and write $I\sim I'$,    if $I-I'$ is divisible by $p$.
Then the set $\mc M/(\sim)$ of equivalence classes  is an $\F_p[z^p]$-module.

\subsection{More general choice of the numbers $M$ and $c$}

For $s=1,\dots,n$ fix a positive integer  $M_s$ such that 
\bea
M_s\equiv -\frac{1}q 
\pmod{p}.
\eea
Fix a positive integer $c'$ such that 
\bea
c' \equiv \frac2q 
\pmod{p}.
\eea
Define the master polynomial, 
\bea
\Phi_p(t,z; \vec M,c') = \prod_{1\leq i<j\leq r}(t_i-t_j)^{c'}\prod_{i=1}^r\prod_{s=1}^n(t_i-z_s)^{M_s}.
\eea
Recall the weight vector-function $W(t,z)$ in \eqref{vW}. The function
$\Phi_p(t,z; \vec M,c') W(t,z)$ is a $W^{\otimes n}[n -2r]$-valued polynomial in $t,z$ with integer coefficients.
Let $l=(l_1,\dots,l_r)\in \Z_{>0}^r$. 
Denote
\bea
I^{(l_1,\dots,l_r)}(z; \vec M,c') = \int_{\{l_1,\dots,l_r\}_p} \Phi_p(t,z; \vec M,c') W(t,z) \,dt_1\dots dt_r\,.
\eea
This is a $W^{\otimes n}[n -2r]$-valued polynomial in $z$ with integer coefficients.

\begin{thm} [\cite{SV2}]
\label{thm sv2n}

For any positive integers $(l_1,\dots,l_r)$, the polynomial
$I^{(l_1,\dots,l_r)}(z; \vec M,c')$ is a solution of the \KZ/ equations modulo $p$ with values in
$\Sing W^{\ox n}[n-2r]$ modulo $p$.

\end{thm}

\begin{thm}
\label{thm non}
For any positive integers $(l_1,\dots,l_r)$, the polynomial
$I^{(l_1,\dots,l_r)}(z; \vec M,c')$ belongs to the module $\mc M$ of $p$-hypergeometric solutions.

\end{thm}

For $l=1$ this statement follows from \cite[Theorem 3.1]{SlV}.

\begin{proof} 

Let $M,c$ be defined in \eqref{Mab}. Then  $c' = c+ d_0p$, $M_s=M+d_sp$, $s=1,\dots,n$,
where $d_0,\dots,d_n$ are nonnegative integers. Then we have modulo $p$,
\bea
\Phi_p(t,z; \vec M,c') \equiv 
\Phi_p(t,z) \prod_{1\leq i<j\leq r}(t_i^p-t_j^p)^{d_0}\prod_{i=1}^r\prod_{s=1}^n(t_i^p-z_s^p)^{d_s}
\pmod{p}.
\eea
This formula implies the theorem.
\end{proof}

\section{Pairs $(p,q)$ of type 1}
\label{sec 3}

\subsection{Numbers $M$ and $c$}

Let $p,q$ be prime numbers, $p>q$. 
Let $k$ be the minimal positive integer such that  $q\vert (k p -1)$.
We have $1\leq k <q$. 

\vsk.2>

We say that the pair $(p,q)$ is of type 1 or type 2  if
\bean
\label{pq}
1\leq k\leq q/2  \qquad\on{or}\qquad q/2 < k < q,
\eean
respectively.

\vsk.2>
The pairs $(p,2)$ are all of type 1.

The pairs $(p,3)$ are of type 1 if $p=3m+1$ and
of type 2 if $p=3m+2$.

\vsk.2>
Let $p=mq+s$,\, $s\in\{1,\dots,q-1\}$. Then 
the minimal positive integer $k$ such that  $q\vert (k p -1)$
belongs to $\{1,\dots,q-1\}$ and is determined by the property
$ks\equiv 1 \pmod{q}$. Hence 
half of the values of $s$ gives pairs
$(p,q)$ of type 1 and 
 half of the values of $s$ gives pairs
$(p,q)$ of type 2.

\vsk.2>
 In the rest of the paper {\it we always assume that $(p,q)$ is of type 1.}

\vsk.2>

\begin{lem}
The integer $q-2k$ is the least positive integer $m$ such that $q\vert (mp+2)$.

\end{lem}

\begin{proof} We have $q\vert (kp-1)$. Hence
$q\vert(qp - 2(kp-1))$, and  $qp - 2(kp-1)= (q-2k)p+2$. We also have $0\leq q-2k \leq q-2$.
\end{proof}

\begin{lem}
The integers 
\bean
\label{M c} 
M=\frac{kp-1}q\,,
\qquad
c =\frac{(q-2k)p+2}q 
\eean
are the least positive integers satisfying the congruences \eqref{Mab}.
The integer $c$ is odd and
\bean
\label{2Mc}
2M+c = p.
\eean
\qed
\end{lem}

\begin{example} 
If $q=2$, then $M=\frac{p-1}2$, $c=1$. The case $q=2$ is the only case in which $c$ does not depend on $p$.

\end{example}

\begin{lem}
\label{lem skew}
For positive integers $(l_1,\dots,l_r)$ 
the polynomial $I^{(l_1,\dots,l_r)}(z)$ equals zero, if $l_1,\dots,l_r$ are not pairwise distinct.
\end{lem}

\begin{proof}
The integer $c$ is odd. Hence the polynomial $\Phi_p(t,z) W(t,z)$ is skew-symmetric with respect to permutations of 
$t_1,\dots, t_r$. 
\end{proof}

 \vsk.2>
  In what follows {\it we always assume}
 \bean
 \label{nqg}
n = qg+2r-1
\eean  
for some positive integer $g$.

\begin{lem}

Let $q, k, g$ be fixed. Let  $p$ be large enough so that 
\bean
\label{mgp}
M-g = \frac{kp-1}q-g =\frac {kp-1-gq}q \geq 0\,.
\eean
Let $l_1  >\cdots > l_r \geq 1$. Then the inequality 
\bean
\label{ne in}
kg+r-1 \geq l_1
\eean
 is necessary for  $I^{(l_1,\dots,l_r)}(z)$ to be nonzero. 
 Hence
there are at most $\binom{kg+r-1}{r}$ tuples $(l_1>\dots>l_r\geq 1)$ such that $I^{(l_1,\dots,l_r)}(z)$ is nonzero.

\end{lem}

\begin{proof}
For $i=1,\dots, r$ we have
\bean
\label{good}
\deg_{t_i} \Phi_p(t,z) W(t,z)
&=& nM + (r-1) c -1
\\
\notag
&=&
(gq+2r-1)\frac{kp-1}q+(r-1)\frac{(q-2k)p+2}q -1
\\
\notag
&=&
(kg+r-1) p - 1 + \frac {kp-1-gq}q  \,.
\eean
This proves the lemma.
\end{proof}

\subsection{Main Theorem}
\label{sec mth}

For a polynomial $F$ in some variables with integer coefficients, denote by 
$[F]_p$ the polynomial  $F$ whose integer coefficients are projected to $\F_p$.

Let 
\bea
f(z)
= \sum_{d_1,\dots,d_n} a_{d_1,\dots,d_n} z_1^{d_1}\dots z_n^{d_n}
\eea
 be a $W^{\ox n}[n-2r]$-valued
 polynomial in $z$ with integer coefficients. Here
 each  $a_{d_1,\dots,d_n}$ is a linear combination of basis vectors $\{ V_J\,|\, J\in\Ik\}$
with integer coefficients.  
Denote  by $[f(z)]_p$ the polynomial 
\bea
[f(z)]_p = \sum_{d_1,\dots,d_n} [a_{d_1,\dots,d_n}]_p \,z_1^{d_1}\dots z_n^{d_n}\,,
\eea
where each  $[a_{d_1,\dots,d_n}]_p$ is the linear combination $a_{d_1,\dots,d_n}$ 
whose integer coefficients are projected to $\F_p$.

Denote by $W^{\ox n}[n-2r]_p$ the vector space over $\F_p$ of  linear combinations of
symbols 
\linebreak
 $\{ V_J\,|\, J\in\Ik\}$
with coefficients in $\F_p$.
Define the subspace $\Sing W^{\ox n}[n-2r]_p \subset W^{\ox n}[n-2r]_p$ as the subspace of all vectors
$\sum_{J\in\Ik} c_JV_J$ whose coefficients  satisfy equations \eqref{ev=0}.

\begin{thm}
\label{thm li}

Let  $(p,q)$ be of type 1.
Let inequality \eqref{mgp} hold. Then  for  any $(l_1,\dots,l_r)$ such that $kg+r-1\geq l_1>\dots>l_r\geq 1$, 
the polynomial $[I^{(l_1,\dots,l_r)}(z)]_p$ is nonzero.
The polynomials
\bea
\{[I^{(l_1,\dots,l_r)}(z)]_p\mid kg+r-1\geq l_1>\dots>l_r\geq 1\}
\eea
are linear independent over $\F_p[z]$, that is, if
\bea
\sum_{kg+r-1 \geq l_1>\dots>l_r\geq 1} f_{l_1,\dots,l_r}(z) [I^{(l_1,\dots,l_r)}(z)]_p =0
\eea
for some $f_{l_1,\dots,l_r}(z) \in\F_p[z]$, then all $f_{l_1,\dots, l_r}(z)$ are equal to zero.

\end{thm}

The theorem is proved in Section \ref{sec proof}.

\begin{cor}
\label{cor rank}

Let  $(p,q)$ be of type 1.
Let inequality \eqref{mgp} hold.  Then the $\F_p[z^p]$-module $\mc M/(\sim)$ is free of rank $\binom{kg+r-1}{r}$
with a basis
$[I^{(l_1,\dots,l_r)}(z)]_p$, where
\linebreak
$kg+r-1\geq l_1>\dots>l_r\geq 1$.
\qed

\end{cor}

\section{Leading term and leading index}
\label{sec 4}

\subsection{Leading term of a polynomial}

Denote by $\succ$  the lexicographical ordering of monomials
$z_1^{d_1}\dots z_n^{d_n}$,\  where $d_1,\dots,d_n\in\Z_{\geq 0}$\,.
  Thus $z_{1} \succ  z_2 \succ \dots \succ z_{n-1}\succ z_n$ and so on.
For example, $z_1^2z_2^2 z_3^2 \succ z_1^2z_2z_4^5$.

For a nonzero polynomial
\bea
f(z)=\sum_{d_1,\dots,d_n} a_{d_1,\dots,d_n} z_1^{d_1}\dots z_n^{d_n}\,
\eea
consider  the summand
$a_{d_1,\dots,d_n} z_1^{d_1}\dots z_n^{d_n}$ corresponding to the
lexicographically  largest
monomial $z_1^{d_1}\dots z_n^{d_n}$ entering   $f(z)$ with a nonzero coefficient
$a_{d_1,\dots,d_n}$. We call this summand
 the leading term of $f(z)$, the corresponding 
$a_{d_1,\dots,d_n}$ --\, the leading coefficient, and
$z_1^{d_1}\dots z_n^{d_n}$ --\, the leading monomial.

\vsk.2>

For example, let $I^{(l_1,\dots,l_r)}(z)$ be one of the
$p$-hypergeometric solutions of Theorem \ref{thm sv2}.
Then 
\bea
[I^{(l_1,\dots,l_r)}(z)]_p = \sum_{d_1,\dots,d_n} [a_{d_1,\dots,d_n}]_p\, z_1^{d_1}\dots z_n^{d_n}\,,
\eea
with  $[a_{d_1,\dots,d_n}]_p\in \Sing W^{\ox n}[n-2r]_p$. 
If $[I^{(l_1,\dots,l_r)}(z)]_p$ is nonzero, then it 
has the leading term, monomial, and coefficient. The leading coefficient is a nonzero vector 
of $\Sing W^{\ox n}[n-2r]_p$.

\subsection{Leading index of a vector of $W^{\ox n}[n-2r]_p$}

We  order $\{ V_J\,|\, J\in\Ik\}$, the basis vectors  of $W^{\ox n}[n-2r]_p$, lexicographically 
with 
the largest of them being  
$V_{\{1,2,\dots,r\}}=w_2\ox \dots \ox w_2\ox w_1\ox \dots \ox w_1$
and the smallest of them being
$V_{\{n-r+1,n-r+2,\dots,n\}}=w_1\ox \dots \ox w_1\ox w_2\ox \dots \ox w_2$.

\vsk.2>

Every nonzero vector $w\in W^{\ox n}[n-2r]_p$ is a linear combination of the basis vectors. 
Let $V_{\{m_1,\dots,m_r\}}$ be the largest of the basis vectors entering $w$  with a nonzero coefficient.  
We call the 
index  $\{m_1,\dots,m_r\}$  the leading index of $w$.

\subsection{Index $\{m_1,\dots,m_r\}$}

Given integers $(l_1,\dots,l_r)$, 
consider the system of inequalities for integers $m_1, \dots, m_r$,
\bean
\label{m}
( m_i-1) M \ \leq  \ nM + (r-i)c-l_ip\  < \  m_i M ,
\qquad i=1,\dots,r .
\eean
Clearly, these inequalities uniquely determine the integers $m_1, \dots, m_r$.

\begin{lem}

\label{lem ine}
Let inequality \eqref{mgp} hold. Let $kg+r-1\geq l_1>\dots>l_r\geq 1$. Then 
  $m_{i} +2\leq m_{i+1}$ for
$i=1,\dots, r-1$, and $1\leq m_1$, $m_r <n$.

\end{lem}

\begin{proof} 

For $i=1,\dots, r-1$ we have 
$(m_i -1)M\leq nM +(r-i)c -l_ip =
nM +(r-i-1)c+c -l_{i+1}p + (l_{i+1}-l_i)p 
\leq 
nM +(r-i-1)c+c -l_{i+1}p -p
=
nM +(r-i-1)c -l_{i+1}p -2M
< (m_{i+1}-2)M$.
Hence $m_{i+1}> m_i+1$. 
 This implies that
$m_{i} +2\leq m_{i+1}$.

The inequality $1\leq m_1$ follows from the inequality
$l_1p \leq nM+(r-1)c$, which is true since
\bea
l_1p \leq (kg+r-1) p  \leq (kg+r-1) p + \frac {kp-1-gq}q = 
nM + (r-1) c,
\eea
see \eqref{good}. We also have  $ M+2c = p \leq l_r p \leq (n-m_r+1)M$.
Hence $m_r < n$.
\end{proof}

\begin{lem}
\label{lem difl}
Let $(l_1,\dots,l_r)$ and $(l_1',\dots,l_r')$ be two distinct tuples of integers
 such that $kg+r-1\geq l_1>\dots>l_r\geq 1$
and $kg+r-1\geq l_1'>\dots>l_r'\geq 1$. Let 
$\{m_1, \dots, m_r\}$ and $\{m_1', \dots, m_r'\}$
be the corresponding sets defined by \eqref{m}.
Then $\{m_1, \dots, m_r\}\ne\{m_1', \dots, m_r'\}$.

\end{lem}

\begin{proof}
Let $l_i>l_i'$ for some $i$. Then
$(m_i -1)M\leq nM +(r-i)c -l_ip =
nM +(r-i)c -l_i'p + (l_i'-l_i)p 
\leq 
nM +(r-i)c -l_i'p -p
=
nM +(r-i)c -l_i'p -2M -c
<
nM +(r-i)c -l_i'p -2M 
< (m_i'-2)M$.
Hence $m_i'> m_i+1$. 
\end{proof}

\subsection{Special summand}

We have
\bea
\notag
\prod_{1\leq i<j\leq r} (t_i-t_j)^c 
&=&
\prod_{1\leq i<j\leq r}
\Big(\sum_{a_{ij}+a_{ji}=c}\binom{c}{a_{ji}}(-1)^{a_{ji}}t_i^{a_{ij}}t_j^{a_{ji}}\Big)
\\
&=&
\sum_{(a_{ij})\in A}\,  t_1^{\sum_{j\ne 1}a_{1j}}\cdots t_r^{\sum_{j\ne r}a_{rj}}
\prod_{i<j} (-1)^{a_{ji}}\binom{c}{a_{ji}},
\eea
where $A=\{(a_{ij})_{1\leq i,j \leq r, \,i\ne j}\mid a_{ij}\in\Z_{\geq 0},\,a_{ij}+a_{ji}=c \ \on{for\, every\,}\, i\ne j\}$.
Hence
\bean
\label{smn}
&&
[I^{(l_1,\dots,l_r)}(z)]_p = [\int_{\{l_1,\dots,l_r\}_p} \Phi_p(t,z) W(t,z) dt_1\dots dt_r]_p\,=
\\
\notag
&&
=\sum_{(a_{ij})\in A}\, \Big(\prod_{i<j} (-1)^{a_{ji}}\binom{c}{a_{ji}}\Big)
[\int_{\{l_1,\dots,l_r\}_p} W(t,z) \prod_{i=1}^r t_i^{\sum_{j\ne i}a_{ij}}\prod_{s=1}^n(t_i-z_s)^M
dt_1\dots dt_r]_p.
\eean
This sum contains the special summand
\bea
S(z):=[\int_{\{l_1,\dots,l_r\}_p} W(t,z) \prod_{i=1}^r t_i^{(r-i)c}\prod_{s=1}^n(t_i-z_s)^Mdt_1\dots dt_r]_p
\eea
corresponding to the collection $(a_{ij})$ such that $a_{ij}=c$ for $i<j$ and $a_{ij}=0$ for $i>j$.

\smallskip
For $1\leq u \leq nM$, denote 
\bean
\label{z(i)}
z(u) : = z_s\,,\qquad \on{if} \quad (s-1)M< u\leq sM.
\eean
Thus $z(1)= z(2)=\dots = z(M) = z_1$, $z(M+1) = z_2$ and so on.

\begin{lem}
\label{lem lea}
The vector-polynomial $S(z)$ is nonzero. The leading monomial of $S(z)$ equals
\bean
\label{lts}
z^{l_1,\dots,l_r}\,:= \prod_{i=1}^r\prod_{u=1}^{nM+(r-i)c-l_ip}z(u).
\eean
Denote by $C_{l_1,\dots,l_r}$ the leading coefficient of $S(z)$. Then 
 the leading index of the vector  $C_{l_1,\dots,l_r}$ equals $\{m_1,\dots,m_r\}$ determined by inequalities \eqref{m}.
Moreover, the leading monomial of any other nonzero summand in \eqref{smn} is lexicographically smaller than
$z^{l_1,\dots,l_r}$.
Thus $C_{l_1,\dots,l_r}z^{l_1,\dots,l_r}$ is the leading term of 
$[I^{(l_1,\dots,l_r)}(z)]_p$.

\end{lem}

\begin{proof}

We rewrite $z^{l_1,\dots,l_r}$ using the integers $m_1,\dots,m_r$,
\bean
\label{lS}
z^{l_1,\dots,l_r} = \prod_{i=1}^r 
z_1^{M}\cdots z_{m_i-1}^{M} \, z_{m_i}^{(n-m_i+1)M + (r-i)c-l_ip}.
\eean
Let  $z^{l_1,\dots,l_r} = z_1^{d_1} \dots z_n^{d_n}$  for some $d_1,\dots,d_n$. Then
\bean
\label{d_j}
d_{m_i}
&=&
(n-m_i+1)M + (r-i)(c+M) -l_ip, \qquad i=1,\dots,r;
\\
\notag
d_j&=&
rM, 
\qquad \qquad \ 
\qquad\qquad\qquad\qquad\qquad
\ j=1,\dots, m_1-1; 
\\
\notag
d_j&=&
(r-i)M, \quad j=m_i+1,\dots, m_{i+1}-1, \quad i=1,\dots, r-1;
\\
\notag
d_j
&=&
0\quad\qquad \qquad \ 
\qquad\qquad\qquad\qquad\qquad\
 \ j=m_r+1,\dots, n.
\eean

It is clear that $z^{l_1,\dots,l_r}$ is the lexicographically maximal monomial which
can be produced by $S(z)$. Let $C_{l_1,\dots,l_r}$ be the coefficient
of $z^{l_1,\dots,l_r}$ in  $S(z)$. Then
$C_{l_1,\dots,l_r}=\sum_{J\in\Ik} c_J V_J$ with $c_J\in\F_p$.
It is clear that $\{m_1,\dots,m_r\}$ is the leading index of $C_{l_1,\dots,l_r}$,
if $c_{\{m_1,\dots,m_r\}}$ is nonzero, but 
\bean
\label{leaco}
c_{\{m_1,\dots,m_r\}} = 
\prod_{i=1}^r(-1)^{(n-m_i+1)M+(r-i)c-l_ip}
\binom{M-1}{(n-m_i+1)M+(r-i)c-l_ip}\,,
\eean
which is nonzero due to \eqref{2Mc} and \eqref{m}.
Thus  $S(z)$ is nonzero, its leading term equals
$C_{l_1,\dots,l_r}z^{l_1,\dots,l_r}$, and the leading index of $C_{l_1,\dots,l_r}$
is $\{m_1,\dots,m_r\}$. 
It remains to prove that the leading monomial of any other
 nonzero summand in \eqref{smn} is lexicographically smaller than
$z^{l_1,\dots,l_r}$.

\vsk.2>

Let $\tilde S(z)$ be any other summand in \eqref{smn},
\bea
\tilde S(z) \,: =\,\Big(\prod_{i<j} (-1)^{a_{ji}}\binom{c}{a_{ji}}\Big)
[\int_{\{l_1,\dots,l_r\}_p} W(t,z) \prod_{i=1}^r t_i^{\sum_{j\ne i}a_{ij}}\prod_{s=1}^n(t_i-z_s)^M
dt_1\dots dt_r]_p.
\eea
It is clear that the  lexicographically maximal monomial which can be produced by $\tilde S(z)$ 
equals
\bean
\label{osum}
\tilde z^{\,l_1,\dots,l_r}\,:=\,\prod_{i=1}^r\prod_{u=1}^{nM-l_ip+\sum_{j\ne i}a_{ij}}z(u).
\eean
Assume that  $\sum_{j\ne 1}a_{1j} < (r-1)c$. Let $s$ be the least index such that the maximal power 
of $z_s$ dividing $\prod_{u=1}^{nM-l_1p+\sum_{j\ne 1}a_{1j}}z(u)$ is strictly 
smaller than the maximal power of $z_s$ (which we denote by $b$)  dividing $\prod_{u=1}^{nM+(r-1)c-l_1p}z(u)$.
Then the maximal power of $z_s$ dividing $\prod_{i=1}^r\prod_{u=1}^{nM-l_ip+\sum_{j\ne i}a_{ij}}z(u)$
is strictly smaller than $b+(r-1)M$. This implies that $\tilde z^{\,l_1,\dots,l_r}$ is lexicographically smaller than $z^{l_1,\dots,l_r}$.

Thus, a summand $\tilde S(z)$ must have  $\sum_{j\ne 1}a_{ij} = (r-1)c$ in order to have a monomial as large lexicographically as 
$z^{l_1,\dots,l_r}$. This means that $a_{1j}=c$ and $a_{j1}=0$  for $j\ne 1$.

\vsk.2>
Now take any summand $\tilde S(z)$ in \eqref{smn} with  $\sum_{j\ne 1}a_{1j} = (r-1)c$. In a similar way we  show that  
$\tilde S(z)$ must have  $\sum_{j\ne 2}a_{2j} = (r-2)c$ in order to have a monomial as large lexicographically as 
$z^{l_1,\dots,l_r}$. Repeating this reasoning we  conclude that the special summand $S(z)$ is the only summand in \eqref{smn}
which may have
the monomial $z^{l_1,\dots,l_r}$ with a nonzero coefficient; and no summands in \eqref{smn}
may have a monomial larger than $z^{l_1,\dots,l_r}$. 
Lemma \ref{lem lea} is proved.  
\end{proof}

\subsection{Proof of Theorem \ref{thm li}}
\label{sec proof}

Let $[I^{(l_1,\dots,l_r)}(z)]_p$  be one of the vector-polynomials of Theorem \ref{thm li}. 
Then  $[I^{(l_1,\dots,l_r)}(z)]_p$ has the leading term $C_{l_1,\dots,l_r}z^{l_1,\dots,l_r}$ described in
Lemma \ref{lem lea}. Also the leading index $\{m_1,\dots,m_r\}$
of $C_{l_1,\dots,l_r}$
is determined by $l_1,\dots,l_r$ in \eqref{m}.
Let $f_{l_1,\dots,l_r}(z) \in\F_p[z]$. Then the leading term of the vector-polynomial
$f_{l_1,\dots,l_r}(z)[I^{(l_1,\dots,l_r)}(z)]_p$  equals the product of leading terms  of
$f_{l_1,\dots,l_r}(z)$ and $[I^{(l_1,\dots,l_r)}(z)]_p$. Moreover,
the leading index of the leading coefficient of 
$f_{l_1,\dots,l_r}(z)[I^{(l_1,\dots,l_r)}(z)]_p$ equals 
the leading index  $\{m_1,\dots,m_r\}$ of the leading coefficient of 
$[I^{(l_1,\dots,l_r)}(z)]_p$.

Consider a
 sum $\sum_{kg+r-1 \geq l_1>\dots>l_r\geq 1} f_{l_1,\dots,l_r}(z) [I^{(l_1,\dots,l_r)}(z)]_p$
 as in Theorem \ref{thm li}.
Then all nonzero summands have 
different   leading indices due to the previous remark and  Lemma \ref{lem difl}. This implies that such a sum is not zero if it has nonzero summands.
Theorem \ref{thm li} is proved.

\begin{example}
Let $q=2$, $n=5$, $r=2$. Then Theorem \ref{thm li} says that for any odd 
prime $p$ there is just one $p$-hypergeometric solution
$[I^{(2,1)}(z_1,\dots,z_5)]_p$. This polynomial is homogeneous of 
degree $2p-4$ and takes values in $\Sing W^{\ox 5}[1]_p$.
The leading term of $[I^{(2,1)}(z_1,\dots,z_5)]_p$ is
$C_{2,1} z_1^{p-2} z_2^{(p-1)/2}z_3^{(p-3)/2}$ where the leading coefficient 
$C_{2,1}\in W^{\ox 5}[1]_p$ has the leading index $\{m_1,m_2\}=\{1,3\}$.

Notice that the space $\Sing W^{\ox 5}[1]$ has dimension 5.

\end{example}

\subsection{Leading terms and eigenvectors}

Let  $I(z) \in \Sing W^{\ox n}[n-2r]_p\ox\F_p[z]$
 be a polynomial solution of the KZ equations \eqref{KZ}
 with a positive integer parameter $q$ 
 (not necessarily a $p$-hypergeometric solution).
 Let $C z_1^{d_1}\dots z_n^{d_n}$ be its leading term,
 $C\in \Sing W^{\ox n}[n-2r]_p$. Then we have modulo $p$,
 \bean
 \label{eig}
 \sum_{\ell=j+1}^n \bar \Om_{j,\ell}\, C \equiv qd_j C, \qquad j=1,\dots,n-1, \qquad d_n\equiv 0,
 \eean
see \cite[Lemma 5.1]{V7}.
 Thus the leading coefficient $C$ is an eigenvector
of the linear operators $\bar \Om_j = \sum_{\ell=j+1}^n \bar \Om_{j,\ell}$, $j=1,\dots,n-1$, with prescribed eigenvalues.
 
 \vsk.2>
 
 An eigenbasis of the operators $\bar \Om_j$, $j=1,\dots,n-1$, on 
 $\Sing W^{\ox n}[n-2r]_p$ is formed by
 the so-called iterated singular vectors, for example see \cite{V2}.
 Such an iterated vector is determined by its leading index.
 
 \vsk.2>
 If $p$ is large enough with respect to $n$, then the vectors of that eigenbasis are separated by eigenvalues.
 
 \vsk.2>

 Let $[I^{(l_1,\dots,l_r)}(z)]_p$ be one of the $p$-hypergeometric solutions of Theorem \ref{thm li}.
 Let 
 \\
 $C_{l_1,\dots,l_r}z_1^{d_1}\dots z_n^{d_n}$ be its leading term, 
 where $d_j$ are defined in \eqref{d_j}. Let $\{m_1,\dots,m_r\}$ be
 the leading index of $C_{l_1,\dots,l_r}$.
  If $p$ is large enough with respect to $n$, then $C_{l_1,\dots,l_r}$ is the iterated singular vector with
  leading index $\{m_1,\dots,m_r\}$. It is the eigenvector of the operators  $\bar \Om_j$, $j=1,\dots,n-1$,
  with eigenvalues defined by formula \eqref{eig}.

\section{Solutions and a Cartier map}
\label{sec 5}

In this section we  discuss how the two objects:
\begin{itemize}
\item the 
 set of indices
 $(l_1,\dots,l_r)$ with $kg+r-1\geq l_1>\dots>l_r\geq 1$,
 appearing in Theorem \ref{thm li}, 
 \item
and the number of such indices
  $\binom{kg+r-1}{r}$, appearing in Corollary \ref{cor rank},
  \end{itemize}
     are related
 to the integrand 
$\Phi(t,z) W(t,z)dt_1 \wedge \dots\wedge dt_r$ of the integral representation
of complex hypergeometric solutions of the \KZ/ equations,
 appearing in Theorems \ref{thm s} and \ref{thm alls}.

\vsk.2>

Recall that the complex hypergeometric solutions
with values in $\Sing W^{\ox n}[n-2r]$  are given by the formulas\,:
\bea
&
I^{(\ga)}(z) = \int_{\ga(z)} \Phi (t,z)  W(t,z)\, dt_1 \wedge \dots\wedge dt_r\,, 
\\
&
\Phi(t,z) = 
\prod_{1\leq i<j\leq r}(t_i-t_j)^{2/q}\prod_{i=1}^r\prod_{s=1}^n(t_i-z_s)^{-1/q},
\\
&
W(t,z) 
= \sum_{1\leq i_1<\dots<i_r \leq n} W_{i_1,\dots,i_r}(t,z)\, V_{i_1,\dots,i_r}\,,
\quad
W_{i_1,\dots,i_r} (t,z)
=
\Sym_{t_1,\dots,t_r} \prod_{j=1}^r\frac 1{t_j-z_{i_j}}\,,
\eea
see Section \ref{Sol in C},
while the corresponding $p$-hypergeometric solutions are given by the formulas\,:
\bea
&
I^{(l_1,\dots,l_r))}(z) = \int_{\{l_1,\dots,l_r\}_p} \Phi_p (t,z)  W(t,z)\, dt_1 \dots dt_r\,, 
\\
&
\Phi_p(t,z) = 
\prod_{1\leq i<j\leq r}(t_i-t_j)^{c}\prod_{i=1}^r\prod_{s=1}^n(t_i-z_s)^{M},
W(t,z)=\sum_{J\in \Ik}W_J(t,z) V_J,
\eea
see Section  \ref{sec 2.5}.

\subsection{Square integrability criterion}

Let $M$ be a complex manifold of complex dimension
$r$.  If $f$ is a meromorphic function on $M$ and $S$
is an irreducible subvariety of $M$, then the order of $f$ along
$S$, $\on{ord}_S(f)$, is the coefficient of the exceptional divisor of the blow up of
$S$ at the divisor of $f$. This notion  generalizes to the setting where $f$
 has only finitely many determinations, which  means that $f$ becomes univalued on a
finite (possibly ramified) cover of $M$. Then $f$ has at a generic point of the
exceptional divisor a fractional order. 

Let $\om$ be a multivalued meromorphic $r$-form on $M$ with
only finitely many determinations and let $S$ be an irreducible subvariety of
$M$.  Write $\om$ at some
point $s$ of $S$ as $f\om_0$ where $\om_0$ is $d$-form on $M$ that is nonzero at $s$ and $f$ is
multivalued meromorphic at $s$. The logarithmic order of $\om$ along $S$ is 
$\on{codim}(S) + \on{ord}_{S;s}(f)$. This only depends on $\om$  and $S$.

Suppose that $D\subset M$ is a hypersurface which is arrangement-like in the sense
that $D$ can be covered by analytic coordinate charts of M on which D is
given by a product of linear forms in the coordinates, see \cite{STV}. It is clear that D then
comes with a natural partition into connected, locally closed submanifolds,
its strata.

We say that a stratum $S$ is decomposable if at a generic point $s\in S$ the germ $D_s$ of the hyperplane
arrangement 
can be decomposed into the disjoint sum $A_1 \cup A_2$ of two germs of hyperplane arrangements and,
after a suitable linear coordinate change, defining equations for $A_1$ and $A_2$ have no common variables.
We say that a stratum $S$ is dense if it is of codimension 1 or if it is not decomposable, see \cite{V1, STV}. 

We have the following trivial observation.

\begin{prop}[\cite{LV}]
\label{prop LV}

 Suppose $M$ is compact and $\om$ is a meromorphic multivalued
$r$-form on $M$ with only finitely many determinations and whose polar set is
contained in an arrangement-like hypersurface $D$. Then the following properties are equivalent:
\begin{enumerate}
\item[(i)]
the form $\om$ is square integrable in the sense that
$\int_M \om\wedge\bar\om$ converges;
\item[(ii)] the form $\om$ has positive logarithmic order along any dense stratum of $D$;
\item[(iii)]
if $\tilde M\to  M$ is a proper surjective map with $\tilde M$ a complex manifold
of the same dimension as $M$ and such that $\om$ becomes a univalued
$d$-form on $\tilde M$, then the latter form is regular.
\end{enumerate}

\end{prop}

\subsection{Square integrable  differential $r$-forms}

Let $t=(t_1,\dots,t_r)$ be coordinates on $\C^r\subset (\Bbb P^1)^r$ and $z=(z_1,\dots,z_n)$ distinct parameters.

Consider a differential $r$-form
\bean
\label{omega}
\om\ &=&\  P(t) \,\Phi(t,z)^k  dt\ 
\\
\notag
&=& \ P(t)\cdot \prod_{i=1}^r\prod_{s=1}^n (t_i-z_s)^{-k/q} 
\cdot\prod_{1\leq i<j\leq n}(t_i-t_j)^{2k/q}\cdot dt_1\wedge\dots\wedge dt_r
\eean
on $\C^r\subset (\Bbb P^1)^r$. Here $k$,  $0<k< q$, is a positive integer; $dt=dt_1\wedge\dots\wedge dt_r$; and $P(t)$ is a polynomial
in $t$. 

\vsk.2>

Let $n=qg+2r-1$ where $g$ is some positive integer, cf. \eqref{nqg}.

\begin{thm}
\label{thm 6.2}
The form $\om$ is square integrable on $(\Bbb P^1)^r$ if and only if 
\bean
\label{deg}
\deg_{t_i}P \leq kg-1 \qquad \on{for\,all}\qquad  i=1,\dots,r.
\eean

\end{thm}

\begin{proof} Denote by $D$ the support of the divisor of $\om$. The support $D$ lies in the union
of hypersurfaces defined by $t_i = z_s$, $t_i =\infty$, and $t_i=t_j$ for $i<j$. This union is
clearly arrangement-like. Its dense strata of codimension $l$ are:
\begin{enumerate}
\item[(i)]
diagonals in $(\Bbb P^1)^r$ defined by letting $l\geq 2$ coordinates to coalesce; 

\item[(ii)] 
loci in $(\Bbb P^1)^r$ defined by setting $l\geq 1$ coordinates equal to $\infty$;

\item[(iii)]
loci in $(\Bbb P^1)^r$ defined by setting $l\geq 1$ coordinates equal to some $z_s$.

\end{enumerate}

Let $S$ be defined by the equations $t_{i_1}=\dots=t_{i_l}$ for some $1\leq i_1<\dots<i_l\leq r$.
Then the logarithmic order of $\om$ along $S$ is\  \ $\geq l-1 + \binom{l}{2}\cdot 2k/q >0$.

Let $S$ be defined by the equations $t_{i_1}=\dots=t_{i_l}=z_s$ for some $1\leq i_1<\dots<i_l\leq r$.
Then the logarithmic order $\om$ along $S$ is 
$\geq \ l -lk/q + \binom{l}{2}\cdot 2k/q = l(1 -k/q) + \binom{l}{2}\cdot 2k/q>0$.

Let $S$ be defined by the equation $t_{i_1}=\infty$ for some $i_1$.
 In the coordinates $u_1,\dots,u_r$, $u_i=1/t_i$ for $i=1,\dots,r$, we have
\bean
\label{om l}
\om &=& (-1)^r \Big( P(1/u_1,\dots, 1/u_r, z) \cdot \prod_{i=1}^r u_r^{\deg_{t_i} P}\Big)
\cdot \prod_{i=1}^r\prod_{s=1}^n (1-z_s u_i)^{-k/q} 
\cdot
\\
\notag
&\cdot&  
\prod_{1\leq i<j\leq n}(u_j-u_i)^{2k/q}\cdot \prod_{i=1}^r u_r^{-\deg_{t_i} P + nk/q-(r-1)2k/q-2} \cdot
du_1\wedge\dots\wedge du_r
\eean
Hence the logarithmic order of $\om$  along $S$ equals 
\bea
1 + nk/q-(r-1)2k/q-2-\deg_{t_{i_1}} P
&=&(n-2(r-1))k/q -1 -\deg_{t_{i_1}}P  
\\
&=&
   kg  - 1 -\deg_{t_{i_1}}P + k/q.
\eea
  Hence the logarithmic order along $S$ is positive
if and only if $\deg_{t_{i_1}}P  \leq kg-1$.

\vsk.2>
Let $S$ be defined by the equations  $t_{i_1}=\dots=t_{i_l}=\infty$ for some $1\leq i_1<\dots<i_l\leq r$.
It follows from \eqref{om l} that
 the logarithmic order along $S$ is positive if the logarithmic order is positive along  every hyperplpane
defined by the equation $t_{i_j}=\infty$ for $j=1,\dots,l$.  The theorem is proved.
\end{proof}

\subsection{Schur polynomials} 
  
  A sequence of integers $a=(a_1\geq a_2\geq \dots\geq a_r\geq 0)$
 is called a partition. For a partition $a$
  the polynomial $m_{a}(t)=\Sym_t t_1^{a_1}\dots  t_r^{a_r}$  
   is   called a  symmetric monomial function.
   The polynomial
\bea
s_{(a_1,\dots,a_r)}(t) = \frac{\det \,(t_i^{a_j +r-j})_{i,j=1,\dots,r}}{\prod_{1\leq i<j\leq r}(t_i-t_j)}
\eea
is called a Schur polynomial.  It is known that
\bean
\label{sm}
s_{a}(t)\ = \ \sum_{b\leq a} \,K_{a,b} \,m_{b_1,\dots,b_r}(t),
\eean
where $K_{a,b}$ are nonnegative integers, $K_{a,a}=1$. The inequality
$b\leq a$ means $\sum_{j=1}^ib_j\leq \sum_{j=1}^ia_j$ for all $i=1,\dots,r$.
The numbers
 $K_{a,b}$ are called the Kostka numbers.

\vsk.2>
For a positive integer $d$, denote
\bea
A(d)=\{(a_1,\dots,a_r)  \mid d\geq a_1\geq a_2\geq \dots\geq a_r\geq 0, \, a_i\in\Z \}.
\eea 
Let $\mc V(d)$ be the free $\Z$-module with basis
$\{ m_a(t) \,|\,a\in A(d) \}$.
The module $\mc V(d)$ has rank $\binom{d+r}{r}$.
The set 
$\{ s_a(t)  \,|\, a\in A(d) \}$
of Schur polynomials is a basis of $\mc V(d)$ by formula \eqref{sm}.

\medskip

Let us return to Theorem \ref{thm 6.2}.
Let
$\mc W$ be the vector space of all differential $r$-forms $\om = P(t)\Phi(t,z)^kdt$ such that $P(t)$ is 
a polynomial in $t_1,\dots,t_r$ symmetric with respect to
permutations of $t_1,\dots,t_r$, and $\om$ is square integrable on  $(\Bbb P^1)^r$.

\begin{cor}
\label{cor sym}

The set 
$\{s_{a}(t) \Phi(t,z)^kdt\, \vert\, a\in A(kg-1)\}$
of differential $r$-forms on $(\Bbb P^1)^r$ is a basis of\  $\mc W$. The vector space  $\mc W$ has  dimension
$\binom{kg+r-1}{r}$.
\qed
\end{cor}
  
  Notice that this binomial coefficient equals  the rank of the module $\mc M/(\sim)$   in Corollary \ref{cor rank}.

\vsk.2>

We introduce the following $\F_p[z]$-variant of the  vector space $\mc W$.
Let $\mc W_p[z]$ be the  free $\F_p[z]$-module with basis 
$\mc B$  of formal algebraic differential $r$-forms 
\bea
\om_a =s_{a}(t) \Phi(t,z)^kdt, \qquad a\in A(kg-1).
\eea
Let  $\mc B^* =\{\om^a \, |\,  a\in A(kg-1)\}$ be the collection of $\F_p[z]$-linear functions on $\mc W_p[z]$
such that $\langle \om^a, \om_b\rangle = \delta_{a,b}$ for all $a,b\in A(kg-1)$.

\subsection{Cartier map}

Assume that $(p,q)$ is of type 1. Hence
$q|(kp-1)$, $0<k\leq q/2$.

\vsk.2>

We define a map which sends every differential form $\Phi(t,z) W_J (t,z)dt $,
$J\in\Ik$,   to $\mc W_p[z]$, that is, to 
a linear combination of  differential $r$-forms $\om_a$, $a\in A(kg-1)$, with coefficients
in $\F_p[z]$. We call this map the Cartier map.
We have
\bea
&&\Phi(t,z) W_J (t,z) dt 
=
 \frac{\Phi(t,z)^{kp}}{\Phi(t,z)^{kp-1}} W_J (t,z) dt
\\
&&=
\Phi(t,z)^{kp} \prod_{1\leq i<j\leq r}(t_i-t_j)^{(2-2kp)/q}\prod_{i=1}^r\prod_{s=1}^n(t_i-z_s)^{(kp-1)/q}
 W_J (t,z) dt
\\
&&=
\frac{\Phi(t,z)^{kp}}{\prod_{1\leq i<j\leq r}(t_i-t_j)^{p}}
\prod_{1\leq i<j\leq r}(t_i-t_j)^{(2+(q-2k)p)/q}\prod_{i=1}^r\prod_{s=1}^n(t_i-z_s)^{(kp-1)/q}
 W_J (t,z) dt
\\
&&
=
\frac{\Phi(t,z)^{kp}}{\prod_{1\leq i<j\leq r}(t_i-t_j)^{p}}
\prod_{1\leq i<j\leq r}(t_i-t_j)^{c}\prod_{i=1}^r\prod_{s=1}^n(t_i-z_s)^{M}
 W_J (t,z) dt
\\
&&
=
\frac{\Phi(t,z)^{kp}}{\prod_{1\leq i<j\leq r}(t_i-t_j)^{p}} \Phi_p(t,z)
 W_J (t,z) dt
\eea
where $c, M$, $\Phi_p(t,z)$ are defined in \eqref{M c} and \eqref{Phi p}.

Since $(p,q)$ is of type 1, the integer $c$ is odd. Hence the  polynomial $\Phi_p(t,z) W_J (t,z) $ is skew-symmetric
 with respect to permutations of $t_1,\dots,t_r$.
We expand $\Phi_p(t,z) W_J (t,z)$ with respect to the $t$-variables as follows:
\bea
&&
\Phi_p(t,z) W_J (t,z)
=
\\
&&
\phantom{a}
=(t_1\dots t_r)^{p-1}
\!\!\!\sum_{(a_1,\dots,a_r)\in A(kg-1)}
\!\!\!
 c_J^{(a_1,\dots,a_r)}(z) \big(\sum_{\si\in S_r}(-1)^{|\si|}
t_{\si(1)}^{(a_1+r-1)p}t_{\si(1)}^{(a_2+r-2)p}\dots
t_{\si(r)}^{a_rp}\big)
+ {\sum}'
\\
&&
\phantom{aaa}
=
(t_1\dots t_r)^{p-1}\!\!\! \sum_{(a_1,\dots,a_r)\in A(kg-1)}\!\!\!
 c_J^{(a_1,\dots,a_r)}(z) s_{(a_1,\dots,a_r)}(t_1^p,\dots,t_r^p) \prod_{1\leq i<j\leq r}(t_i^p-t_j^p) + {\sum}',
\eea
where  $\sum'$ denotes the sum of the monomials  $t_1^{d_1}\dots,t_r^{d_r}$ such that at least one of
$d_1,\dots,d_r$ is not of the form $lp-1$ for some positive integer $l$; the coefficients $c_J^{(a_1,\dots,a_r)}(z) $ are suitable polynomials in $z$.

Returning to $\Phi(t,z) W_J (t,z) dt$ we write
\bea
&
\Phi(t,z) W_J (t,z) dt 
=
\\
&
=
\sum_{(a_1,\dots,a_r)\in A(kg-1)}
 c_J^{(a_1,\dots,a_r)}(z) s_{(a_1,\dots,a_r)}(t_1^p,\dots,t_r^p)\prod_{1\leq i<j\leq r} \frac{ t_i^p-t_j^p}{(t_i-t_j)^p}
  \Phi(t,z)^{kp}(t_1\dots t_r)^{p-1}dt
\\
&+\ \ 
 \frac{ \sum'}{\prod_{1\leq i<j\leq r}(t_i-t_j)^{p}} \Phi(t,z)^{kp} dt\,.
\eea
Notice that $\frac{ t_i^p-t_j^p}{(t_i-t_j)^p}\equiv 1$ (mod $p$).  
We define the {\it Cartier map} $\mc C$ by the formula
\bean
\label{def Cart}
&&
\mc C\  :\  \Phi(t,z) W_J (t,z) dt  \ \mapsto 
\!\!\!
\sum_{(a_1,\dots,a_r)\in A(kg-1)}
\!
[c_J^{(a_1,\dots,a_r)}(z)]_p\, s_{(a_1,\dots,a_r)}(t_1,\dots,t_r) \Phi(t,z)^{k} dt \, ,
\eean
cf. \cite{AH}.

\vsk.2>

Recall the collection $\{\om^{(a_1,\dots, a_r)} \,|\, (a_1,\dots, a_r)\in A(kg-1)\}$
of linear functions on $\mc W_p[z]$. We have
\bea
\big\langle \om^{(a_1,\dots, a_r)}, \,\mc C (\Phi(t,z) W_J (t,z) dt)\big\rangle = [c_J^{(a_1,\dots, a_r)}(z)]_p\,.
\eea

\vsk.2>

\begin{thm}
\label{thm Cart}

For any $p$-hypergeometric solution $[I^{(a_1+r, a_2+r-1,\dots,a_r+1)}(z)]_p$,
$kg-1\geq a_1\geq a_2\geq\dots \geq a_r\geq 0$ we have
\bean
\label{Ca}
[I^{(a_1+r, a_2+r-1,\dots,a_r+1)}(z)]_p\  =\ \sum_{J\in\Ik}\
[c_J^{(a_1,\dots,a_r)}(z)]_p V_J\,.
\eean

\end{thm}

\begin{proof}  The proof follows from the definition of 
$I^{(a_1+r, a_2+r-1,\dots,a_r+1)}(z)$ in \eqref{def I}.
\end{proof}

Formula \eqref{Ca} can be reformulates as follows. For any integers
$(l_1,\dots,l_r)$, $kg+r-1\geq l_1>\dots>l_r\geq 1$,
the $p$-hypergeometric solution 
 $[I^{(l_1,\dots,l_r)}(z)]_p$ is given by the formula
\bean
\label{phC} 
[I^{(l_1,\dots,l_r)}(z)]_p =
\left\langle \om^{(l_1-r,\, l_2-r+1,\,\dots,\,l_r-1)}, \,\mc C (\Phi(t,z) W(t,z) dt)\right\rangle.
\eean

\vsk.4>
\noindent
Notice that  $\Phi(t,z) W(t,z) dt$ is the integrand of the integral representation of complex
hypergeometric solutions of the \KZ/ equations, see 
\eqref{intrep}, while $[I^{(l_1,\dots,l_r)}(z)]_p$ is a solution 
of the \KZ/ equations over $\F_p$. 

\vsk.2>
Formula \eqref{phC} for $r=1$ and two prime numbers $(p>q)$ not necessarily of type 1
is the subject of \cite[Theorem 6.2]{SlV}.

\end{document}